\pdfoutput=1  
\documentclass[letterpaper, 10pt, conference]{ieeeconf}
\IEEEoverridecommandlockouts
\usepackage{amsmath,amssymb}

\usepackage{amsthm}
\usepackage{graphicx}
\usepackage{booktabs}
\usepackage{xcolor}
\usepackage{algorithm}
\usepackage{algorithmic}
\usepackage[hidelinks]{hyperref}

\newtheorem{proposition}{Proposition}
\newtheorem{definition}{Definition}
\newtheorem{corollary}{Corollary}

\newcommand{\C}{\mathbb{C}}
\newcommand{\R}{\mathbb{R}}
\newcommand{\argu}{\operatorname{arg}}
\newcommand{\KSW}{\mathrm{KSW}}
\newcommand{\SP}{\mathrm{SP}}
\newcommand{\Dsec}{\mathrm{Det}}

\title{\LARGE \bf
A Spectral--Phase Admissibility Certificate for Complex Linear Maps}

\author{Snigdha Chandan Khilar$^{1}$%
\thanks{$^{1}$Independent Researcher.
        {\tt\small snkhilar@gmail.com}}%
}

\begin{document}
\maketitle
\thispagestyle{empty}
\pagestyle{empty}

\begin{abstract}
The Kontsevich--Segal--Witten (KSW) criterion, $\sum_i |\argu \lambda_i| < \pi$,
characterizes which complex metrics yield convergent Gaussian path integrals in
quantum gravity. We import this criterion into machine learning as a constraint
on the spectrum of a complex linear map, and study what it does and does not
provide. We first show, and prove, that the KSW value is \emph{orthogonal} to
every magnitude-based constraint used for stability (spectral norm, orthogonal
and unitary parameterizations) and \emph{distinct} from positive-definiteness:
it constrains the collective \emph{phase} of the spectrum, an axis these methods
leave free. We then establish an exact, brute-force--verified hierarchy of three
differentiable, $O(d^3)$ certificates of increasing strength---a determinant
sector, a subset-product envelope $\SP$, and full KSW---and prove which one
certifies which family of derived quantities. The subset-product certificate is
shown to be the exact condition under which \emph{every} exterior-power
eigenvalue of the map stays off the negative real axis, i.e.\ under which an
exponentially large family of log-volume readouts remains well defined; it makes
the identical accept/reject decision as an exponential minor enumeration at up to
$10^{9}\times$ lower cost. We give a differentiable enforcement based on a Schur
parameterization and demonstrate, on an objective that genuinely depends on the
exterior-power tower, that the subset-product certificate is the correct choice:
a determinant-only guard is fooled, full KSW is over-conservative, and brute
force is redundant. Finally, we delimit scope with two
results: KSW is \emph{not} a stabilizer of deep linear propagation (tightening
the budget worsens eigenvector conditioning), and magnitude-based objectives
---normalizing-flow likelihoods and positive-definite log-determinants---are
provably phase-blind and hence unaffected by the certificate. Its applicability
is therefore confined to objectives that consume the argument of a spectral
product. All code is available at
\url{https://github.com/nssprogrammer/ksw}.
\end{abstract}

\section{INTRODUCTION}

Several successful neural architectures are, at bottom, imports of a physical
structure: diffusion models from nonequilibrium
thermodynamics~\cite{sohldickstein,ddpm}, Hopfield
networks and attention from associative-memory spin systems~\cite{hopfield,attention},
and neural ODEs from dynamical systems~\cite{neuralode}. Motivated by this pattern, we examine a criterion that
has, to our knowledge, no prior machine-learning use: the
Kontsevich--Segal--Witten (KSW) admissibility criterion for complex
metrics~\cite{ks,witten}, recently emphasized in the study of gravitational
wormholes and complex saddle points~\cite{mmm}.

For a complex metric with eigenvalues $\lambda_i$ (relative to a fixed positive
reference), KSW declares the metric \emph{admissible}---the Gaussian path
integral of every $p$-form fluctuation converges---if and only if
\begin{equation}
\sum_{i=1}^{d} \bigl|\argu \lambda_i\bigr| < \pi ,
\label{eq:ksw}
\end{equation}
where $\argu$ denotes the principal argument. The criterion constrains the
collective \emph{phase} of a spectrum. Intuitively, each eigenvalue is a complex
number with a size $|\lambda_i|$ and an angle $\argu\lambda_i$; \eqref{eq:ksw}
caps the \emph{total angle} the spectrum may occupy, which keeps every
eigenvalue---and every product of eigenvalues---off the negative real axis, so
that determinants and their logarithms remain single-valued. This is conspicuous
because essentially
every stability tool in machine learning constrains a spectrum's
\emph{magnitude}: spectral normalization bounds the largest singular value,
orthogonal and unitary parameterizations fix $|\lambda_i| = 1$, and Lipschitz
regularizers bound gain. The phase axis is left free. It is therefore natural to
ask whether \eqref{eq:ksw} is a useful, and genuinely new, constraint.

\paragraph*{Contributions and honest scope}
We answer this in three parts, and we are deliberate about the boundary between
what the criterion does and does not deliver.
\begin{enumerate}
\item \textbf{Distinctness (Sec.~\ref{sec:distinct}, \ref{sec:distinct-exp}).}
We prove and verify that the KSW value is orthogonal to every magnitude-based
constraint and distinct from positive-definiteness. It occupies unoccupied
ground: the collective phase of the spectrum.
\item \textbf{An exact certificate hierarchy (Sec.~\ref{sec:theory}).}
We show that \eqref{eq:ksw} sits atop a hierarchy of three differentiable,
$O(d^3)$ certificates---a determinant sector $\Dsec$, a subset-product envelope
$\SP$, and full $\KSW$, with $\Dsec \le \SP \le \KSW \le 2\SP$---and we prove
which family of derived quantities each certifies. We give elementary proofs and
confirm them against exponential brute force.
\item \textbf{Cost and use (Sec.~\ref{sec:algo}, \ref{sec:exp}).}
We give a differentiable enforcement, show the certificate reproduces an
exponential minor enumeration's exact decision at up to $10^9\times$ lower cost,
and demonstrate on a tower-dependent objective that the subset-product
certificate is the correct one.
\end{enumerate}
We also report a negative result that fixes the criterion's scope: KSW is
\emph{not} a stabilizer of deep linear propagation. Under iterated maps,
tightening the budget clusters eigenvalues and \emph{worsens} conditioning
(Sec.~\ref{sec:negative}). The criterion's value lies not in generic layers but
in objectives that touch the exterior-power tower---determinants, subdeterminant
families, and multi-scale log-volumes.

\section{BACKGROUND: THE KSW CRITERION}
\label{sec:background}

We recall the origin of \eqref{eq:ksw} because it dictates the correct
machine-learning translation; the physics is used only as motivation, and the
remainder of the paper (Sec.~\ref{sec:theory} onward) is self-contained and
requires none of it. Consider a complex metric diagonalized as
$g = \operatorname{diag}(\lambda_1,\dots,\lambda_d)$ and a $p$-form field with
kinetic action $\int \sqrt{g}\, g^{-1}\cdots g^{-1} F\wedge \star F$. A mode
whose indices split the directions into a subset $S$ (``lowered'') and its
complement (``raised'') carries a Gaussian coefficient
\begin{equation}
c_\varepsilon \;=\; \prod_{i} \lambda_i^{\varepsilon_i/2},
\qquad \varepsilon_i \in \{+1,-1\},
\label{eq:sector}
\end{equation}
with $\varepsilon_i = -1$ for $i \in S$. The integral over that mode converges
iff $\operatorname{Re}(c_\varepsilon) > 0$, i.e.\ iff
$\bigl|\tfrac12 \sum_i \varepsilon_i \argu\lambda_i\bigr| < \tfrac{\pi}{2}$,
equivalently
\begin{equation}
\Bigl| \sum_i \varepsilon_i \argu \lambda_i \Bigr| < \pi .
\label{eq:persector}
\end{equation}
Requiring \eqref{eq:persector} for \emph{all} $2^d$ sign patterns
$\varepsilon$ is exactly \eqref{eq:ksw}, as we prove below. The criterion is thus
an envelope: a single scalar certifying an exponential family of convergences.
This is the property we will exploit, and the reason a faithful port must target
objects built from products over subsets of the spectrum.

\section{THEORY: A CERTIFICATE HIERARCHY}
\label{sec:theory}

\begin{definition}[Spectral phase functionals]
Let $M \in \C^{d\times d}$ have eigenvalues $\lambda_1,\dots,\lambda_d$ and set
$a_i = \argu\lambda_i \in (-\pi,\pi]$. With
$\Sigma_+ = \sum_{a_i>0} a_i$ and $\Sigma_- = \sum_{a_i<0} |a_i|$, define
\begin{align*}
\Dsec(M) &= \Bigl|\textstyle\sum_i a_i\Bigr|, \\
\SP(M)   &= \max(\Sigma_+,\Sigma_-), \\
\KSW(M)  &= \textstyle\sum_i |a_i|.
\end{align*}
\end{definition}

In words, $\Dsec$ is the phase of the determinant, $\SP$ is the largest phase any
product over a subset of eigenvalues can reach, and $\KSW$ is the total phase
budget of the spectrum; all three ignore eigenvalue magnitudes entirely.

\begin{proposition}[Sign-pattern envelope]
\label{prop:ksw}
$\displaystyle \KSW(M) = \max_{\varepsilon \in \{\pm1\}^d} \Bigl|\sum_i \varepsilon_i a_i\Bigr|.$
\end{proposition}
\begin{proof}
Taking $\varepsilon_i=\operatorname{sign}(a_i)$ gives $\sum_i \varepsilon_i a_i
= \sum_i |a_i|$, so the maximum is $\ge \KSW(M)$. Conversely
$|\sum_i \varepsilon_i a_i| \le \sum_i |\varepsilon_i a_i| = \sum_i|a_i|$ for
every $\varepsilon$. \end{proof}

\begin{corollary}[KSW $=$ all-sector convergence]
\label{cor:sectors}
$\KSW(M)<\pi$ iff every $p$-form sector \eqref{eq:persector} converges.
\end{corollary}

\begin{proposition}[Subset envelope]
\label{prop:sp}
$\displaystyle \SP(M) = \max_{S \subseteq \{1,\dots,d\}} \Bigl|\sum_{i\in S} a_i\Bigr|.$
\end{proposition}
\begin{proof}
For any $S$, $\sum_{i\in S} a_i \le \sum_{i\in S,\,a_i>0} a_i \le \Sigma_+$ and
$\sum_{i\in S} a_i \ge -\Sigma_-$, so $|\sum_{i\in S}a_i|\le \max(\Sigma_+,\Sigma_-)$.
Equality holds at $S=\{i:a_i>0\}$ (value $\Sigma_+$) or $S=\{i:a_i<0\}$
(value $\Sigma_-$). \end{proof}

\begin{proposition}[Exact subset-product admissibility]
\label{prop:exact}
Every subset product $\prod_{i\in S}\lambda_i$, carried with its additive phase
$\sum_{i\in S} a_i$, stays off the negative real axis
$(|\sum_{i\in S}a_i|<\pi)$ for all $S$ if and only if $\SP(M)<\pi$.
\end{proposition}
\begin{proof}
Immediate from Proposition~\ref{prop:sp}: the worst subset attains $\SP(M)$.
\end{proof}

The eigenvalues of the $k$-th exterior power $\Lambda^k(M)$ are exactly the
subset products $\prod_{i\in S}\lambda_i$ with $|S|=k$. Thus
Proposition~\ref{prop:exact} states that $\SP(M)<\pi$ is the exact condition for
every exterior-power eigenvalue---hence every $\log\det \Lambda^k$ and every
multi-scale log-volume---to remain well defined (no branch-cut crossing).

\begin{proposition}[Hierarchy]
\label{prop:hier}
$\Dsec(M) \le \SP(M) \le \KSW(M) \le 2\,\SP(M).$
\end{proposition}
\begin{proof}
$\Dsec=|\Sigma_+-\Sigma_-|\le \max(\Sigma_+,\Sigma_-)=\SP$;
$\SP=\max(\Sigma_+,\Sigma_-)\le \Sigma_++\Sigma_-=\KSW$;
$\KSW=\Sigma_++\Sigma_-\le 2\max(\Sigma_+,\Sigma_-)=2\SP$. \end{proof}

Propositions~\ref{prop:ksw}--\ref{prop:hier} give three certificates of strictly
increasing strength, each computable from one eigendecomposition:
$\Dsec<\pi$ controls $\det M$ alone; $\SP<\pi$ controls the entire exterior-power
tower and is \emph{exact} for it; $\KSW<\pi$ controls the $p$-form ratio sectors
and is the strongest. The gap $\SP<\pi\le\KSW$ is not vacuous: a spectrum can
have every subset product admissible while a ratio sector diverges
(Sec.~\ref{sec:exp}, Fig.~\ref{fig:hier}).

\subsection{Distinctness from magnitude and positivity}
\label{sec:distinct}

\begin{proposition}[Phase, not magnitude]
\label{prop:orth}
$\KSW$, $\SP$, and $\Dsec$ are invariant under $\lambda_i \mapsto r_i \lambda_i$
for any $r_i>0$. Consequently they are independent of all singular values and of
the spectral radius; every magnitude-based constraint leaves them unconstrained,
and conversely.
\end{proposition}
\begin{proof}
$\argu(r_i\lambda_i)=\argu\lambda_i$ for $r_i>0$, so all three functionals are
unchanged, while singular values and $|\lambda_i|$ can be set arbitrarily. \end{proof}

Positive-definiteness is also distinct: requiring $\operatorname{Re}\lambda_i>0$
for all $i$ (each $|a_i|<\pi/2$) neither implies nor is implied by $\KSW<\pi$.
Three eigenvalues at $a_i=0.4\pi$ satisfy $\operatorname{Re}\lambda_i>0$ yet give
$\KSW=1.2\pi>\pi$; conversely an eigenvalue at $a=0.9\pi$
($\operatorname{Re}\lambda<0$) paired with one at $-0.05\pi$ gives
$\KSW=0.95\pi<\pi$. We confirm both directions empirically in
Sec.~\ref{sec:distinct-exp}.

\section{ALGORITHM}
\label{sec:algo}

\subsection{A differentiable admissible layer}
To train under any of the three certificates without differentiating through a
general (ill-conditioned) eigendecomposition, we parameterize the map in Schur
form, exposing the eigenphases directly. Let $A\in\C^{d\times d}$ be free and set
$Q=\exp(A-A^{\mathsf H})$ (unitary), the exponential-map parameterization used
for orthogonal and unitary constraints~\cite{expmparam}. Let
$T$ be upper triangular with diagonal $r_i e^{\mathrm i\theta_i}$
($r_i=e^{\rho_i}>0$) and free strictly-upper entries (non-normality). Then
\begin{equation}
W = Q\,T\,Q^{\mathsf H}, \qquad \argu\lambda_i(W)=\theta_i,
\label{eq:schur}
\end{equation}
since $W$ and $T$ are unitarily similar and $T$ is triangular. The certificates
become closed-form functions of $\theta$: $\Dsec=|\sum_i\theta_i|$,
$\SP=\max(\Sigma_+,\Sigma_-)$, $\KSW=\sum_i|\theta_i|$---no eigensolver in the
forward pass.

\subsection{Enforcement}
We enforce a budget $b<\pi$ by radial projection of the phase vector, or by a
one-sided penalty $\max(0,\,\mathcal C(\theta)-b)$ added to the loss, where
$\mathcal C\in\{\Dsec,\SP,\KSW\}$. Projection is exact and parameter-free
(Algorithm~\ref{alg:proj}); the penalty is preferred when the constraint should
be soft. Because all three functionals are $1$-homogeneous in $\theta$, a single
scaling restores feasibility.

\begin{algorithm}[t]
\caption{Phase projection onto an admissibility ball}
\label{alg:proj}
\begin{algorithmic}[1]
\REQUIRE eigenphases $\theta\in\R^d$; certificate $\mathcal C\in\{\Dsec,\SP,\KSW\}$; budget $b<\pi$
\STATE $\Sigma_+ \leftarrow \sum_{\theta_i>0}\theta_i$;\quad $\Sigma_- \leftarrow \sum_{\theta_i<0}|\theta_i|$
\STATE $v \leftarrow |\Sigma_+-\Sigma_-|,\ \max(\Sigma_+,\Sigma_-),\ \Sigma_++\Sigma_-$ for $\mathcal C=\Dsec,\SP,\KSW$
\IF{$v>b$}
  \STATE $\theta \leftarrow \theta \cdot b/v$ \COMMENT{$1$-homogeneous rescale}
\ENDIF
\RETURN $\theta$
\end{algorithmic}
\end{algorithm}

\section{EXPERIMENTS}
\label{sec:exp}

All quantities are computed in double precision; code accompanies the paper.
We verify each theoretical claim against exponential brute force before using it.

\subsection{Distinctness from existing constraints}
\label{sec:distinct-exp}

To isolate phase from magnitude we sample normal matrices with all
$|\lambda_i|=1$ (spectral radius exactly one, so every magnitude-based criterion
is blind) and random phases. Over $2{\times}10^4$ trials at $d=4$, ``spectral
radius $<1$'' accepts $0\%$ while $\KSW<\pi$ still accepts $4.4\%$, discriminating
purely on phase---an empirical counterpart to Proposition~\ref{prop:orth}. The
two positivity witnesses of Sec.~\ref{sec:distinct} are reproduced exactly:
the state with $\operatorname{Re}\lambda_i>0$ everywhere yet $\KSW=1.2\pi$ is
rejected, and the state with a negative-real-part eigenvalue yet $\KSW=0.95\pi$
is accepted. KSW is therefore neither a magnitude nor a positive-definiteness
constraint.

\subsection{The hierarchy, verified against brute force}
Over $1500$ random spectra at $d=8$ we confirm
Propositions~\ref{prop:ksw}--\ref{prop:hier} to machine precision: the closed
forms for $\SP$ and $\KSW$ match the maxima over all $2^d$ subsets and sign
patterns, the ordering $\Dsec\le\SP\le\KSW\le2\SP$ always holds, and
$\SP<\pi$ coincides exactly with ``all subset products off the negative real
axis'' (Proposition~\ref{prop:exact}). The gap is common: $287/1500\approx19\%$
of spectra satisfy $\SP<\pi\le\KSW$---the whole exterior-power tower is
admissible while a $p$-form ratio sector diverges. The nesting is drawn in
Fig.~\ref{fig:hier}.

\begin{figure}[t]
\centering
\includegraphics[width=0.86\columnwidth]{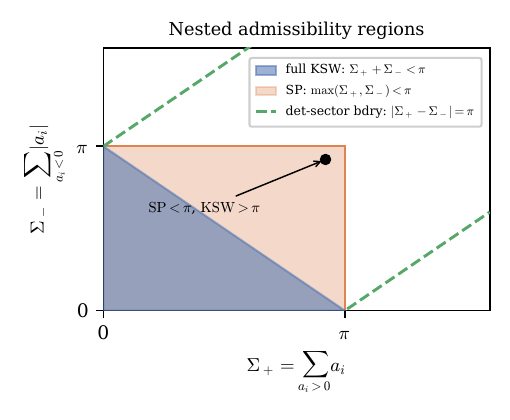}
\caption{The three certificates as regions in the $(\Sigma_+,\Sigma_-)$ plane,
where $\Sigma_+$ ($\Sigma_-$) is the total positive (negative) phase of the
spectrum. A matrix maps to one point here and is admissible under a certificate
iff its point lies in that region. Full KSW (triangle) $\subseteq$ subset-product
$\SP$ (square) $\subseteq$ the determinant-sector band. The marked point has the
whole exterior-power tower admissible ($\SP<\pi$) yet violates full KSW.}
\label{fig:hier}
\end{figure}

\subsection{Computational cost}
The certificate is an eigendecomposition plus an $O(d)$ reduction; the honest
alternative---checking that every subset product is admissible---enumerates
$2^d$ subsets. Fig.~\ref{fig:cost} reports both, together with a correctness
guard: across all sizes the $O(d^3)$ certificate returns the \emph{identical}
accept/reject decision as the exponential enumeration (Proposition~\ref{prop:exact}).
At $d=18$ brute force already costs $1.2\,$s per decision versus $0.27\,$ms for
the certificate ($\sim\!4500\times$); the projected gap at $d=40$ exceeds
$10^{9}\times$. Under the Schur parameterization the eigendecomposition is
avoided entirely and the certificate is $O(d)$.

\begin{figure}[t]
\centering
\includegraphics[width=0.86\columnwidth]{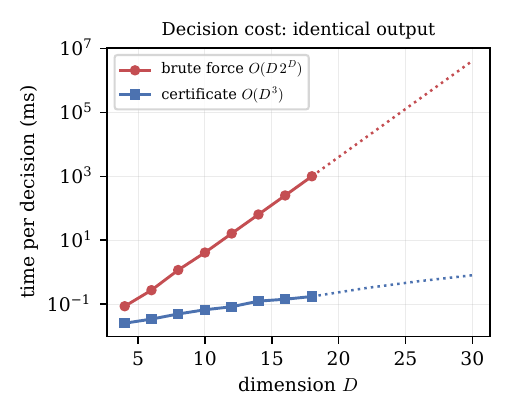}
\caption{Per-decision cost. The $O(d^3)$ certificate and the $O(d\,2^d)$
enumeration return identical decisions at all $d$; dotted lines extrapolate.}
\label{fig:cost}
\end{figure}

\subsection{Application: a tower-dependent objective}
\label{sec:app}
We construct the setting in which the criterion is meant to help: an objective
that reads out multi-scale log-volumes, i.e.\ the principal complex logarithm of
subset products $\prod_{i\in S}\lambda_i$. When a subset phase
$\sum_{i\in S}\theta_i$ crosses $\pm\pi$, that readout jumps by $2\pi$,
introducing a discontinuity. We drive a learnable spectrum toward a
\emph{mixed} target ($d=8$: four phases at $+0.55\pi$, four at $-0.55\pi$), so
that $\sum_i\theta_i\approx0$---the determinant is fooled---while the four
positive phases sum to $2.2\pi$, forcing a subset crossing. Table~\ref{tab:app}
compares five guards under identical training.

\begin{table}[t]
\caption{Constraining a tower-dependent objective ($d{=}8$, budget $0.9\pi$).
``cross.'' counts branch-cut crossings during training; ``fit'' is the residual
to the target (lower is more expressive).}
\label{tab:app}
\centering
\begin{tabular}{lcccl}
\toprule
guard & cross. & fit & $\SP/\pi$ & outcome \\
\midrule
none        & 1 & 0.80 & 2.03 & crosses \\
det-sector  & 1 & 0.80 & 2.03 & fooled ($\sum\theta{\approx}0$) \\
$\SP$       & 0 & 9.28 & 0.90 & correct, best feasible \\
$\KSW$      & 0 & 15.1 & 0.45 & over-conservative \\
brute (minors) & 0 & 9.28 & 0.90 & $\equiv\SP$, $2^d$ cost \\
\bottomrule
\end{tabular}
\end{table}

The reading is unambiguous. The determinant guard is inactive (its value
$\approx0$) and permits the crossing---guarding $\det$ alone is unsafe for the
tower. $\SP$ prevents every crossing at the best feasible fit. $\KSW$ also
prevents crossings but from a strictly smaller ball, paying a large fit penalty
($15.1$ vs.\ $9.28$)---quantifying the cost of the physics-strongest certificate
when the model needs only the subset tower. Brute-force minor projection is
byte-identical to $\SP$ at exponential cost. Thus $\SP$ is the correct tool;
$\Dsec$ too weak, $\KSW$ too strong, brute force redundant.

\subsection{Negative result: not a propagation stabilizer}
\label{sec:negative}
Because \eqref{eq:ksw} is phase-only (Proposition~\ref{prop:orth}), it cannot
control the magnitude growth of an iterated map $x_{t+1}=Wx_t$; that is governed
by singular values and, at fixed magnitude, by eigenvector conditioning
$\kappa(V)$. Fig.~\ref{fig:cond} makes the consequence explicit: at fixed
spectral radius one, \emph{tightening} the phase budget clusters the eigenvalues
near a common angle, driving $\kappa(V)$---and hence transient amplification and
gradient growth---up by many orders of magnitude. The gradient factor through
$100$ steps falls from $2.7{\times}10^5$ at budget $0.1\pi$ to $32$ at $3\pi$.
The naive reading of KSW as a deep-network stabilizer is therefore false; if
anything a tight budget is harmful. This delimits the criterion's scope to static
admissibility of determinantal and exterior-power objectives, exactly where
Sec.~\ref{sec:app} shows it to be the right instrument.

\begin{figure}[t]
\centering
\includegraphics[width=0.86\columnwidth]{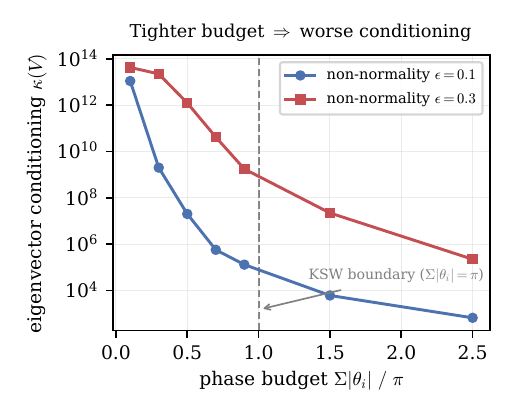}
\caption{Negative result. At fixed unit spectral radius, a tighter phase budget
clusters eigenvalues and \emph{worsens} eigenvector conditioning $\kappa(V)$,
increasing transient/gradient amplification under iteration. KSW does not
stabilize propagation.}
\label{fig:cond}
\end{figure}

\section{SCOPE AND APPLICABILITY}
\label{sec:scope}
The certificate acts on the \emph{phase} of a spectrum, so it is inert for any
objective that reads the spectrum only through magnitudes. This is not a soft
caveat but a theorem, and it delimits applicability sharply.

\begin{proposition}[Magnitude objectives are phase-blind]
\label{prop:blind}
Let $W\in\C^{k\times k}$ act on $\C^k\cong\R^{2k}$ with real representation
$R(W)=\left[\begin{smallmatrix}\operatorname{Re}W & -\operatorname{Im}W\\
\operatorname{Im}W & \operatorname{Re}W\end{smallmatrix}\right]$. Then
$\det_{\R} R(W)=|\det_{\C}W|^2$, so
$\log|\det_{\R}R(W)| = 2\sum_i \log|\lambda_i(W)|$, which is invariant under
$\Dsec,\SP,\KSW$. More generally, any loss depending on a matrix spectrum only
through the magnitudes $|\lambda_i|$---including the normalizing-flow
log-likelihood $\log|\det J|$ and the log-determinant of a positive-definite
kernel (real, positive spectrum)---is unchanged by the certificates and can be
neither improved nor harmed by them.
\end{proposition}
\begin{proof}
$\det_{\R}R(W)=|\det_{\C}W|^2$ is standard. Hence
$\log|\det_{\R}R(W)|=2\operatorname{Re}\sum_i\log\lambda_i=2\sum_i\log|\lambda_i|$.
The functionals $\Dsec,\SP,\KSW$ depend only on $\argu\lambda_i$, which is absent
here; for a positive-definite spectrum $\argu\lambda_i=0$. \end{proof}

Proposition~\ref{prop:blind} pins down the applicable regime. The certificate
matters only when the loss consumes the \emph{argument} of a spectral product,
i.e.\ an explicit complex log-volume $\log\prod_{i\in S}\lambda_i$ whose imaginary
part enters the objective, as in Sec.~\ref{sec:app}. A survey of the dominant
determinant and spectral objectives in machine learning finds none of this form:
normalizing-flow likelihoods use $|\det J|$~\cite{glow,nfsurvey}; Gaussian
processes, determinantal point processes, and Gaussian graphical models use
log-determinants of positive-definite kernels~\cite{gp,dpp,ndpp}; complex-valued
networks place the loss on output magnitude~\cite{cvnn}. All are phase-blind by
Proposition~\ref{prop:blind}. Complex determinants do appear natively in lattice
field theory with a sign problem~\cite{signproblem}, but there the determinant
phase is the physical signal to be estimated, not a well-posedness nuisance to be
constrained, so the certificate is intent-mismatched. We therefore claim the
certificate as a complete, cheap tool for a specific and presently uncommon
objective class---explicit complex log-volume readouts---and not as a drop-in for
mainstream generative or kernel models.

\section{RELATED WORK}
Magnitude-based spectral control is standard: spectral
normalization~\cite{specnorm}, Parseval and orthogonal networks~\cite{parseval},
and unitary or complex recurrent maps~\cite{urnn,wisdom,trabelsi} constrain singular values or fix
$|\lambda|=1$; antisymmetric and Lipschitz recurrent networks~\cite{antirnn,lipschitzrnn}
obtain norm-preservation from skew-symmetric generators. All operate on the
magnitude axis and, by Proposition~\ref{prop:orth}, leave the KSW phase functional
free. Non-normality, transient growth, and the conditioning of deep linear
propagation are classically analyzed through pseudospectra and dynamical
isometry~\cite{pseudo,saxe,dyniso}, consistent with
Sec.~\ref{sec:negative}. On the application side, log-determinant objectives and
(sub)determinant positivity underlie Gaussian processes~\cite{gp}, determinantal
point processes~\cite{dpp,ndpp}, and normalizing flows~\cite{realnvp,glow,nfsurvey}; by
Proposition~\ref{prop:blind} these are phase-blind, so our certificate is
complementary rather than competing. Its contribution is a single, cheap,
differentiable certificate for the well-posedness of the exterior-power tower,
imported from complex-metric admissibility~\cite{ks,witten,mmm} and, to our
knowledge, new to machine learning.

\section{DISCUSSION AND LIMITATIONS}
We have shown that the KSW criterion ports faithfully to a spectral-phase
certificate that is provably distinct from existing (magnitude, positivity)
constraints, that it heads an exact and cheap hierarchy whose middle member
$\SP$ is the correct certificate for exterior-power well-posedness, and that its
scope is determinantal rather than dynamical. Two points bound the present
claims. First, applicability is narrow and we have delimited it precisely rather
than left it open: Proposition~\ref{prop:blind} (Sec.~\ref{sec:scope}) proves
that magnitude-based objectives---normalizing-flow likelihoods and
positive-definite log-determinants alike---are phase-blind, so the certificate
applies only to the presently uncommon class of objectives that consume the
argument of a spectral product. The demonstration in Sec.~\ref{sec:app} is
accordingly a mechanism study on such an objective, not a benefit claim on a
mainstream model, and we make no such claim. Second, KSW is defined for
a symmetric metric; applying it to a general non-normal map is a modeling choice
we made explicit, and the Schur parameterization sidesteps the numerical
fragility of eigen-phases for strongly non-normal matrices. We regard the result
as a small, self-contained tool with a clear usage rule---use $\SP$ for the
exterior-power tower, $\Dsec$ for determinants alone, and neither for propagation
stability.

\end{document}